\documentclass{article}
\usepackage[dvipdfmx]{graphicx}
\usepackage{latexsym}
\usepackage{amsfonts}
\usepackage{amsmath}
\usepackage{amsthm}
\usepackage{mathrsfs}

\newcommand{\bcr}{{\rm bcr}}

\newtheorem{lemma}{Lemma}
\newtheorem{theorem}{Theorem}

\title{A faster fixed parameter algorithm for two-layer crossing minimization}
\date{}
\author{Yasuaki Kobayashi \and Hisao Tamaki}
\begin{document}
\maketitle

\begin{abstract}
We give an algorithm that decides whether the bipartite crossing number
of a given graph is at most $k$.
The running time of the algorithm is upper bounded by $2^{O(k)} + n^{O(1)}$,
where $n$ is the number of vertices of the input graph,
which improves the previously known algorithm due to Kobayashi {\it et al.} (TCS 2014)
that runs in $2^{O(k \log k)} + n^{O(1)}$ time.
This result is based on a combinatorial upper bound on the number of two-layer drawings
of a connected bipartite graph with a bounded crossing number.
\end{abstract}

\section{Introduction}
\label{sec:intro}

A {\em two-layer drawing} of a bipartite graph is a drawing in which
the vertices in one color class are placed on a straight line, the vertices in
the other color class are placed on another straight line parallel to the first line,
and each edge is drawn as a straight line segment.
The two parallel lines are called {\em layers}.
A {\em crossing} in a two-layer drawing is a pair of edges that intersect each other
at a point distinct from their end point.
The problem of finding a two-layer drawing with the minimum number of crossings,
called {\em two-layer crossing minimization} (or simply TLCM),
is dealt with as a combinatorial problem:
the number of crossings in a two-layer drawing is determined by the order
of vertices on each layer.

TLCM is shown to be NP-hard \cite{GJ82}.
It is worth mentioning that the original proof in \cite{GJ82} shows, in fact, the hardness
of TLCM for multigraphs,
however, Schaefer \cite{Sch13} shows that TLCM is NP-hard even for simple graphs.
TLCM can be solved in polynomial time for trees \cite{SSSV01}
and for bipartite permutation graphs \cite{SBS87}.
TLCM has been studied from the view point of parameterized complexity.
In this context, we are asked if there is a two-layer drawing
of a given bipartite graph with at most $k$ crossings,
where $k$ is the parameter for the parameterized problem.
Dujimovi\'c {\it et al.} show that this parameterized
problem is fixed parameter tractable \cite{DFK+08}.
More generally, they give an algorithm that decides whether a given graph
with $n$ vertices has an $h$-layer drawing with at most $k$ crossings
in $2^{O((h + k)^3)}n$ time (see \cite{DFK+08}, for details).
The current and two other authors \cite{KNM+14} improve the running time
for the restricted case $h = 2$, namely TLCM, to $2^{O(k \log k)} + n^{O(1)}$.
Moreover, they, for the first time, show that TLCM admits a polynomial kernelization.
In this paper, we give a faster fixed parameter algorithm for TLCM.

\begin{theorem}
  \label{thm:main}
  There is an algorithm that decides whether a given bipartite graph has
  a two-layer drawing with at most $k$ crossings whose running time is $2^{O(k)} + n^{O(1)}$,
  where $n$ is the number of vertices of the input graph.
\end{theorem}

To establish Theorem~\ref{thm:main}, we analyze the number of two-layer drawings
whose crossing number is at most $k$ and enumerate all such drawings in the claimed running time.
This strategy is inspired by the work of Gutin {\it et al.} \cite{GRS+07}.
They consider a parameterized version of the linear arrangement problem and
give a fixed parameter algorithm for the problem.
To this end, they analyze the number of feasible solutions
of the problem for a spanning tree of the input graph.
However, the approach of using spanning trees does not seem to work for
our problem and we need a slightly more involved technique using Eulerian tours.
See Section~\ref{sec:upper_bound}.

\section{Preliminaries}
\label{sec:pre}
Let $G$ be a bipartite graph with a prescribed bipartition, denoted by $(X(G), Y(G))$, of the vertex set.
We denote by $E(G) \subseteq X(G) \times Y(G)$ the set of edges of $G$.
We call a vertex of degree one a {\em leaf} and the edge incident to a leaf a {\em leaf edge}.

For a set $S$, a {\em layout} on $S$ is a bijection $f$ from $S$ to $\{1, 2, \ldots, |S|\}$.
A {\em two-layer drawing} $D$ of $G$ is defined by
a triple $(G, f_X, f_Y)$, where $f_X$ and $f_Y$
are layouts on $X(G)$ and $Y(G)$, respectively.
A {\em crossing} in $D$ is a pair of edges $(x, y)$ and $(x', y')$
such that both $f_X(x) < f_X(x')$ and $f_Y(y') < f_Y(y)$ hold.
The {\em crossing number} of $D$, denoted by $\bcr(D)$, is
the number of crossings in $D$, that is,
\begin{equation*}
  \sum_{(x, y) \in E(G)} |\{(x', y') \in E(G) : f_X(x) < f_X(x'), f_Y(y') < f_Y(y)\}|.
\end{equation*}
The {\em bipartite crossing number} $\bcr(G)$ of $G$ is the minimum $k$
such that there exists a two-layer drawing $D$ of $G$ with $\bcr(D) = k$.

\section{Combinatorial upper bound}\label{sec:upper_bound}
In this section, fix a connected bipartite graph $G$ and an integer $k$.
We give an upper bound on the number of two-layer drawings of $G$ with at most $k$ crossings.

Let $n = |X(G) \cup Y(G)|$.
One may notice that a trivial upper bound $(n - 1)!$ is essentially tight since
$K_{1, n - 1}$ has $(n - 1)!$ different two-layer drawings without any crossings.
One would, however, also notice that all of those drawings are equivalent in a natural sense.
More generally, it is straightforward to verify that,
for every bipartite graph $G$, there is an optimal two-layer drawing of $G$
in which all the leaves adjacent to $v$ appear consecutively in their layer
for each non-leaf vertex $v$.
Therefore, we may treat such sibling leaves as a single leaf,
weighting the corresponding leaf edge by the number of represented leaves \cite{KNM+14}.
We call a pair of leaves a {\em sibling pair}, if they have a common neighbor.

\begin{lemma}
  \label{lem:upper_bound}
  Suppose $G$ has no sibling pairs.
  Then, the number of two-layer drawings of $G$ with at most $k$ crossings
  is $2^{O(n + k)}$.
\end{lemma}
\begin{proof}
  The lemma is trivial when $|X(G)| = 1$ or $|Y(G)| = 1$. Hence we assume otherwise.
  Choose $r \in X(G)$ arbitrarily and define a function $T: X(G) \setminus \{r\} \rightarrow X(G)$ that
  satisfies the following conditions:
  \begin{enumerate}
  \item for each $x \in X(G) \setminus \{r\}$, there is a path $P_x$ of length
    two between $x$ and $T(x)$ in $G$,
  \item for each $e \in E(G)$, $|\{x \in X(G) \setminus \{r\} : e \in E(P_x)\}| \le 2$, and
  \item $(X(G), \{\{x, T(x)\} : x \in X(G) \setminus \{r\}\})$ forms a tree.
  \end{enumerate}
  Such a function $T$ is defined as follows.
  Let $H$ be a bipartite multigraph obtained from $G$ by
  replacing each edge by two parallel edges.
  Since every vertex in $H$ has even degree and $H$ is connected, $H$ has an Eulerian tour.
  Then, we fix an Eulerian tour starting at $r$.
  For $x \in X(G) \setminus \{r\}$, we let $T(x)$
  be the vertex in $X(G)$ that is visited immediately after the last visiting of $x$ in the tour.
  It is straightforward to verify that the function $T$ satisfies condition 1 and 2.
  Since $(X(G), \{\{x, T(x)\} : x \in X(G) \setminus \{r\}\})$ has $|X(G)| - 1$ edges
  and, by the construction of $T$, has no cycles, condition 3 holds.
  
  Fix a two-layer drawing $D = (G, f_X, f_Y)$.
  For $x \in X(G) \setminus \{r\}$, we define $g(x) = |f_X(x) - f_X(T(x))| - 1$.
  In other words, $g(x)$ is the number of vertices of $X(G)$ that lie between $x$ and $T(x)$ in $D$.
  Since $G$ has no sibling pairs, there is at most one leaf in $X(G)$ adjacent to
  the unique vertex in $V(P_x) \cap Y(G)$.
  Observe that each edge incident to a vertex counted by $g(x)$ except
  for such a leaf (if it exists) makes a crossing with an edge of $P_x$.
  Considering double counts and the fact that each edge belongs
  to at most two paths $P_x$, $x \in X(G)$, we have
  \begin{equation*}
    \frac{1}{4}\sum_{x \in X(G)\setminus \{r\}} (g(x) - 1) \le k.
  \end{equation*}
  Therefore, the number of possible functions $g$, which is non-negative, has domain size
  at most $n - 1$, and has total value of at most $4k + n - 1$, is at most
  \begin{equation*}
    \left(
    \begin{array}{c}
      4k+2n-3\\
      4k+n-1
    \end{array}
    \right) \le 2^{4k+2n-3}.
  \end{equation*}
  By condition 3, $f_X$ is determined by the values of $g(x)$ and the signs of $|f_X(x) - f_X(T(x))|$ for each $x \in X(G) \setminus \{r\}$
  and the value of $f_X(r)$.
  Hence the number of possible layouts $f_X$ is bounded by $2^{4k+2n-3} \cdot 2^{n - 1} \cdot n$.
  Applying the same argument to $Y(G)$ proves the lemma.
\end{proof}

The above proof immediately gives an algorithm that enumerates
all two-layer drawings of $G$ with at most $k$ crossings in $2^{O(n + k)}$ time
when $G$ is connected and has no sibling pairs.

\section{FPT algorithm}
\label{sec:alg}
In this section, we will design a fixed parameter algorithm for TLCM
using the enumeration algorithm in the previous section.
Our fixed parameter algorithm uses a kernelization result due to \cite{KNM+14}.
First, we define an weighted version of TLCM as follows.
Consider a two-layer drawing of an edge weighted bipartite graph.
We assume that, in this paper, each edge has weight at least one.
The {\em weight} of a crossing is the product of the weights of the crossing edges.
The crossing number of the drawing
is the sum of the weights of all crossings in the drawing.
The bipartite crossing number of an edge weighted bipartite graph is defined analogously.
A {\em leaf edge weighted graph} is an edge weighted graph, where each non-leaf edge has weight exactly one.
The following result is needed for our fixed parameter algorithm.

\begin{theorem}[\cite{KNM+14}]
  \label{thm:kernel}
  There is a polynomial time algorithm that,
  given a connected bipartite graph $G$ and an integer $k$,
  either computes a leaf edge weighted connected bipartite graph $H$ such that
  $H$ has no sibling pairs, $|E(H)| = O(k)$, and $\bcr(G) = \bcr(H)$ or
  answers $\bcr(G) > k$.
\end{theorem}

Given a bipartite graph and an integer $k$,
we apply the algorithm in Theorem~\ref{thm:kernel} to each connected component of $G$.
If one of the applications answers the bipartite crossing number of
a connected component is larger than $k$, we obviously conclude that $\bcr(G) > k$.
For each output, by Lemma~\ref{lem:upper_bound},
the number of two-layer drawings with at most $k$ crossings is bounded by $2^{O(k)}$.
The bipartite crossing number of each output can be computed in $2^{O(k)}$ time
by using the enumeration algorithm in the previous section
and hence Theorem~\ref{thm:main} holds.

\section{Conclusion}
In this paper, we have developed a faster FPT algorithm for TLCM.
This improves the previous running time $2^{O(k \log k)} + n^{O(1)}$ \cite{KNM+14}
to $2^{O(k)} + n^{O(1)}$.
It would be interesting to know the existence of a subexponential time
FPT algorithm, namely $2^{o(k)}n^{O(1)}$ time algorithm.
The reductions of Garey and Johnson \cite{GJ82} and Schaefer \cite{Sch13}
combined with Exponential Time Hypothesis \cite{IP01}
do not seem to imply a reasonable complexity lower bound for this question.

\end{document}